\documentclass[11pt, dvipsnames, letterpaper]{article}

\usepackage{amsmath, amssymb, amsthm}
\usepackage{mathtools}
\usepackage{fullpage}
\usepackage{graphics}
\usepackage{environ}
\usepackage{framed}
\usepackage{url}
\usepackage{algorithm}
\usepackage[noend]{algpseudocode}
\usepackage[labelfont=bf]{caption}
\usepackage{cite}
\usepackage{framed}
\usepackage[framemethod=tikz]{mdframed}
\usepackage{appendix}
\usepackage{graphicx}
\usepackage[textsize=tiny]{todonotes}
\usepackage{tcolorbox}
\usepackage{hyperref}
\usepackage{multirow}

\def\ShowAuthNotes{1}
\ifnum\ShowAuthNotes=1
\newcommand{\authnote}[2]{\ \\ \textcolor{red}{\parbox{0.9\linewidth}{[{\footnotesize {\bf #1:} { {#2}}}]}}\newline}
\else
\newcommand{\authnote}[2]{}
\fi

\hypersetup{
pagebackref,
letterpaper=true,
colorlinks=true,
urlcolor=blue,
linkcolor=blue,
citecolor=OliveGreen,
}

\newcommand\remove[1]{}

\theoremstyle{plain}
\newtheorem{lemma}{Lemma}[section]
\newtheorem*{lemma*}{Lemma}
\newtheorem{corollary}[lemma]{Corollary}
\newtheorem*{corollary*}{Corollary}

\newtheorem{claim}{Claim}

\theoremstyle{definition}
\newtheorem{theorem}[lemma]{Theorem}

\newtheorem*{theorem*}{Theorem}
\newtheorem{definition}[lemma]{Definition}

\newtheorem*{rem*}{Remark}

\usepackage{prettyref}
\newcommand{\savehyperref}[2]{\texorpdfstring{\hyperref[#1]{#2}}{#2}}

\newrefformat{eq}{\savehyperref{#1}{\textup{(\ref*{#1})}}}
\newrefformat{lem}{\savehyperref{#1}{Lemma~\ref*{#1}}}
\newrefformat{def}{\savehyperref{#1}{Definition~\ref*{#1}}}
\newrefformat{thm}{\savehyperref{#1}{Theorem~\ref*{#1}}}
\newrefformat{cor}{\savehyperref{#1}{Corollary~\ref*{#1}}}
\newrefformat{cha}{\savehyperref{#1}{Chapter~\ref*{#1}}}
\newrefformat{sec}{\savehyperref{#1}{Section~\ref*{#1}}}
\newrefformat{app}{\savehyperref{#1}{Appendix~\ref*{#1}}}
\newrefformat{tab}{\savehyperref{#1}{Table~\ref*{#1}}}
\newrefformat{fig}{\savehyperref{#1}{Figure~\ref*{#1}}}
\newrefformat{hyp}{\savehyperref{#1}{Hypothesis~\ref*{#1}}}
\newrefformat{alg}{\savehyperref{#1}{Algorithm~\ref*{#1}}}
\newrefformat{rem}{\savehyperref{#1}{Remark~\ref*{#1}}}
\newrefformat{item}{\savehyperref{#1}{Item~\ref*{#1}}}
\newrefformat{step}{\savehyperref{#1}{step~\ref*{#1}}}
\newrefformat{conj}{\savehyperref{#1}{Conjecture~\ref*{#1}}}
\newrefformat{fact}{\savehyperref{#1}{Fact~\ref*{#1}}}
\newrefformat{prop}{\savehyperref{#1}{Proposition~\ref*{#1}}}
\newrefformat{prob}{\savehyperref{#1}{Problem~\ref*{#1}}}
\newrefformat{claim}{\savehyperref{#1}{Claim~\ref*{#1}}}
\newrefformat{relax}{\savehyperref{#1}{Relaxation~\ref*{#1}}}
\newrefformat{rem}{\savehyperref{#1}{Remark~\ref*{#1}}}
\newrefformat{red}{\savehyperref{#1}{Reduction~\ref*{#1}}}
\newrefformat{part}{\savehyperref{#1}{Part~\ref*{#1}}}
\newrefformat{ex}{\savehyperref{#1}{Exercise~\ref*{#1}}}
\newrefformat{property}{\savehyperref{#1}{Property~\ref*{#1}}}
\newcommand{\Sref}[1]{\hyperref[#1]{\S\ref*{#1}}}

\let\pref=\prettyref

\DeclarePairedDelimiter\floor{\lfloor}{\rfloor}

\newcommand\R{\mathbb{R}}

\renewcommand{\bf}{\textbf}
\newcommand{\eps}{\varepsilon}
\newcommand{\bS}{\boldsymbol{S}}
\newcommand{\bx}{\boldsymbol{x}}
\newcommand{\by}{\boldsymbol{y}}
\newcommand{\bw}{\boldsymbol{w}}

\newcommand{\calA}{\mathcal{A}}
\newcommand{\calD}{\mathcal{D}}

\newif\ifrandom
\randomtrue

\newcommand{\poly}{{\rm poly}}

\newcommand{\FindRep}{\textsc{Find-Duplicate}}
\newcommand{\OWFindRep}{\textsc{One-Way-Find-Duplicate}}
\newcommand{\OWPartRec}{\textsc{One-Way-Partial-Recovery}}
\newcommand{\FindSupElem}{\textsc{Find-Support-Elem}}
\newcommand{\calB}{\mathcal{B}}
\newcommand{\NonzeroRow}{\textsc{Find-Nonzero-Row}}
\newcommand{\wt}{\widetilde}
\newcommand{\Equality}{\textsc{Equality}}
\newcommand{\RecBasis}{\textsc{RecoverBasis}}

\title{Pseudo-deterministic Streaming}
\author{}
 \author{Shafi Goldwasser\thanks{Supported by NSF CNS-1413920,  DARPA/NJIT  491512803, Sloan  Foundation 996698, and MIT/IBM W1771646. This work was done at the Simons Institute for the Theory of Computing.} \and Ofer Grossman\thanks{Supported by the Fannie and John Hertz Foundation fellowship, an NSF GRFP award, NSF CNS-1413920,  DARPA/NJIT  491512803, Sloan  Foundation 996698, and MIT/IBM W1771646. This work was done in part at the Simons Institute for the Theory of Computing.} \and Sidhanth Mohanty\thanks{EECS Department, University of California Berkeley.  Supported by NSF grant CCF-1718695} \and David P. Woodruff\thanks{Supported by the National Science Foundation under Grant No. CCF-1815840. This work was done in part at the Simons Institute for the Theory of Computing.}}

\date{\today}

\begin{document}

\maketitle

\begin{abstract}
A pseudo-deterministic algorithm is a (randomized) algorithm which, when run multiple times on the same input, with high probability outputs the same result on all executions. Classic streaming algorithms, such as those for finding heavy hitters, approximate counting, $\ell_2$ approximation, finding a nonzero entry in a vector (for turnstile algorithms) are not pseudo-deterministic. For example, in the instance of finding a nonzero entry in a vector, for any known low-space algorithm $A$, there exists a stream $x$ so that running $A$ twice on $x$ (using different randomness) would with high probability result in two different entries as the output.

In this work, we study whether it is inherent that these algorithms output different values on different executions. That is, we ask whether these problems have low-memory pseudo-deterministic algorithms. For instance, we show that there is no low-memory pseudo-deterministic algorithm for finding a nonzero entry in a vector (given in a turnstile fashion), and also that there is no low-dimensional pseudo-deterministic sketching algorithm for $\ell_2$ norm estimation.  We also exhibit problems which do have low memory pseudo-deterministic algorithms but no low memory deterministic algorithm, such as outputting a nonzero row of a matrix, or outputting a basis for the row-span of a matrix.

We also investigate multi-pseudo-deterministic algorithms: algorithms which with high probability output one of a few options. We show the first lower bounds for such algorithms. This implies that there are streaming problems such that every low space algorithm for the problem must have inputs where there are many valid outputs, all with a significant probability of being outputted.
\end{abstract}

\section{Introduction}

Consider some classic streaming problems: heavy hitters, approximate counting, $\ell_p$ approximation, finding a nonzero entry in a vector (for turnstile algorithms), counting the number of distinct elements in a stream. These problems were shown to have low-space randomized algorithms in \cite{charikar2004finding,morrisapxcounting,flajolet1985approximate,alon1999space,indyk2005optimal,monemizadeh20101}, respectively. All of these algorithms exhibit the property that when running the algorithm multiple times on the same stream, different outputs may result on the different executions.

For the sake of concreteness, let's consider the problem of $\ell_2$ approximation: given a stream of poly($n$) updates to a vector (the vector begins as the zero vector, and updates are of the form ``increase the $i^{\text{th}}$ entry by $1$'' or ``decrease the $j^{\text{th}}$ entry by $1$''), output an approximation of the $\ell_2$ norm of the vector. There exists a celebrated randomized algorithm for this problem \cite{alon1999space}. This algorithm has the curious property that running the same algorithm multiple times on the same stream may result in different approximations. That is, if Alice runs the algorithm on the same stream as Bob (but using different randomness), Alice may get some approximation of the $\ell_2$ norm (such as 27839.8), and Bob (running the same algorithm, but with your own randomness) may get a different approximation (such as 27840.2). The randomized algorithm has the guarantee that both of the approximations will be close to the true value. However, interestingly, Alice and Bob end up with slightly different approximations. Is this behavior inherent? That is, could there exist an algorithm which, while being randomized, for all streams with high probability both Alice and Bob will end up with the \textit{same} approximation for the $\ell_2$ norm?

Such an algorithm, which when run on the same stream multiple times outputs the same output with high probability is called \textit{pseudo-deterministic}. The main question we tackle in this paper is:

\begin{center}
\emph{What streaming problems have low-memory pseudo-deterministic algorithms?}
\end{center}



\subsection{Our Contributions}

This paper is the first to investigate pseudo-determinism in the context of streaming algorithms. We show certain problems have pseudo-deterministic algorithms substantially faster than the optimal deterministic algorithm, while other problems do not.

\subsubsection{Lower Bounds} 
\paragraph{\FindSupElem:} 
We show pseudo-deterministic lower bounds for finding a nonzero entry in a vector in the turnstile model. Specifically, consider the problem \FindSupElem~of finding a nonzero entry in a vector for a turnstile algorithm (the input is a stream of updates of the form ``increase entry $i$ by 1'' or ``decrease entry $j$ by 1'', and we wish to find a nonzero entry in the final vector). We show this problem does not have a low-memory pseudo-deterministic algorithm:


\begin{theorem}\label{psdlb}
There is no pseudo-deterministic algorithm for  \FindSupElem~which uses $\tilde{o}(n)$ memory.
\end{theorem}

This is in contrast with the work of \cite{monemizadeh20101}, which shows a randomized algorithm for the problem using polylogarithmic space.

Theorem \ref{psdlb} can be viewed as showing that any low-memory algorithm $A$ for \FindSupElem~must have an input $x$ where the output $A(x)$ (viewed as a random variable depending on the randomness used by $A$) must have at least a little bit of entropy. The algorithms we know for \FindSupElem~have a very high amount of entropy in their outputs (the standard algorithms, for an input which is the all 1s vector, will find a uniformly random entry). Is this inherent, or can the entropy of the output be reduced? We show that this is inherent:
 for every low memory algorithm there is an input $x$ such that $A(x)$ has high entropy. 


\begin{theorem}
Every randomized algorithm for \FindSupElem~using $o(s)$ space must have an input $x$ such that $A(x)$ has entropy at least $\log\left(\frac{n}{s\log n}\right)$.
\end{theorem}

So, in particular, an algorithm using $n^{1- \eps}$ space must have outputs with entropy $\Omega(\log n)$, which is maximal up to constant factors.

We also show analogous lower bounds for the problem \FindRep in which the input is a stream of $3n/2$ integers between $1$ and $n$, and the goal is to output a number $k$ which appears at least twice in the stream:

\begin{theorem}
Every randomized algorithm for \FindRep~using $o(s)$ space must have an input $x$ such that $A(x)$ has entropy at least $\log\left(\frac{n}{s\log n}\right)$.
\end{theorem}

\paragraph{Techniques}

To prove a pseudo-deterministic lower bound for \FindSupElem, the idea is to show that if a pseudo-deterministic algorithm existed for \FindSupElem, then there would also exist a pseudo-deterministic one-way communication protocol for the problem \OWFindRep, where Alice has a subset of $[n]$ of size $3n/4$, and so does Bob, and they wish to find an element which they share.

To prove a lower bound on the one-way communication problem \OWFindRep, we show that if such a pseudo-deterministic protocol existed, then Bob can use Alice's message to recover many ($n/10$) elements of her input (which contains much more information than one short message).  The idea is that using Alice's message, Bob can find an element they have in common. Then, he can remove the element he found that they have in common from his input, and repeat to find another element they have in common (using the original message Alice sent, so Alice does not have to send another message). After repeating $n/10$ times, he will have found many elements which Alice has.

It may not be immediately obvious where pseudo-determinism is being used in this proof. The idea is that because the algorithm is pseudo-deterministic, the element which Bob finds as the intersection with high probability does not depend on the randomness used by Alice. That is, let $b_1, b_2, \ldots$ be the sequence of elements which Bob finds. Because the algorithm is pseudo-deterministic, there exists a specific sequence $b_1, b_2, \ldots$ such that with high probability this will be the sequence of elements which Bob finds. Notice that a randomized (but not pseudo-deterministic) algorithm for \OWFindRep would result in different sequences on different executions.

When the sequence $b_1, b_2, \ldots$ is determined in advance, we can use a union bound and argue that with high probability, one of Alice's messages will likely work on all of Bob's inputs. If $b_1, b_2, \ldots$  is not determined in advance, then it's not possible to use a union bound.

Proving a lower bound on the entropy of the output of an algorithm for \FindSupElem~uses a similar idea, but is more technically involved. It is harder to ensure that Bob's later inputs will be able to succeed with Alice's original message. The idea, at a very high level, is to have Alice send many messages (but not too many), so that Bob's new inputs will not strongly depend on any part of Alice's randomness, and also to have Alice send additional messages to keep Bob from going down a path where Alice's messages will no longer work.

This lower bound technique may seem similar to the way one would show a deterministic lower bound. It's worth noting that for certain problems, deterministic lower bounds do not generalize to pseudo-deterministic lower bounds; see our results on pseudo-deterministic upper bounds for some examples and intuition for why certain problems remain hard in the pseudo-deterministic setting while others do not.

\paragraph{Sketching lower bounds for pseudo-deterministic $\ell_2$ norm estimation:}  The known randomized algorithms (such as \cite{alon1999space}) for approximating the $\ell_2$ norm of a vector $x$ in a stream rely on \emph{sketching}, i.e., storing $\bS x$ where $\bS$ is a $d\times n$ random matrix where $d \ll n$ and outputting the $\ell_2$ norm of $\bS x$.  More generally, an abstraction of this framework is the setting where one has a distribution over matrices $\calD$ and a function $f$.  One then stores a \emph{sketch} of the input vector $\bS x$ where $\bS\sim\calD$ and outputs $f(\bS x)$.  By far, most streaming algorithms fall into this framework and in fact some recent work \cite{li2014turnstile,DBLP:conf/coco/AiHLW16} proves under some caveats and assumptions that low-space turnstile streaming algorithms imply algorithms based on low-dimensional sketches.  Since sketching-based streaming algorithms are provably optimal in many settings, it motivates studying whether there are low-dimensional sketches of $x$ from which the $\ell_2$ norm can be estimated pseudo-deterministically.

We prove a lower bound on the dimension of sketches from which the $\ell_2$ norm can be estimated pseudo-deterministically: 

\begin{theorem}
    Suppose $\calD$ is a distribution over $d\times n$ matrices and $f$ is a function from $\R^d$ to $\R$ such that for all $x\in\R^n$, when $\bS\sim\calD$:
    \begin{itemize}
    \item $f(\bS x)$ approximates the $\ell_2$ norm of $x$ to within a constant factor with high probability,
    \item $f(\bS x)$ takes a unique value with high probability.
    \end{itemize}
    Then $d$ must be $\Omega\left(n\right)$.
\end{theorem}

As an extension, we also show that
\begin{theorem}
    For every constant $\eps,\delta>0$, every randomized sketching algorithm $A$ for $\ell_2$ norm estimation using a $O(n^{1-\delta})$-dimensional sketch, there is a vector $x$ such that the output entropy of $A(x)$ is at least $1-\eps$.  Furthermore, there is a randomized algorithm using a $O(\poly\log n)$-dimensional sketch with output entropy at most $1+\eps$ on all input vectors.
\end{theorem}

\paragraph{Techniques} The first insight in our lower bound is that if there is a pseudo-deterministic streaming algorithm $A$ for $\ell_2$ norm estimation in $k$ space, then that means there is a fixed function $g$ such that $g(x)$ approximates $\|x\|_2$ and $A$ is a randomized algorithm to compute $g(x)$ with high probability.  The next step uses a result in the work of \cite{hardt2013robust} to illustrate a (randomized) sequence of vectors $\bx^{(1)},\dots,\bx^{(t)}$ only depending on $g$ such that any linear sketching-based algorithm that uses sublinear dimensional sketches outputs an incorrect approximation to the $\ell_2$ norm of some vector in that sequence with constant probability, thereby implying a dimension lower bound.

\subsubsection{Upper Bounds}

On the one hand, all the problems considered so far were such that
\begin{enumerate}
    \item There were ``low-space'' randomized algorithms.
    \item The pseudo-deterministic and deterministic space complexity were ``high'' and equal up to logarithmic factors.
\end{enumerate}

This raises the question if there are natural problems where pseudo-deterministic algorithms outperform deterministic algorithms (by more than logarithmic factors). We answer this question in the affirmative.

We illustrate several natural problems where the pseudo-deterministic space complexity is strictly smaller than the deterministic space complexity.

The first problem is that of finding a nonzero row in a matrix given as input in a turnstile stream.  Our result for this problem has the bonus of giving a natural problem where the pseudo-deterministic streaming space complexity is strictly sandwiched between the deterministic and randomized streaming space complexity.

In the problem \NonzeroRow, the input is an $n \times d$ matrix $A$ streamed in the turnstile model, and the goal is to output an $i$ such that the $i^{th}$ row of the matrix $A$ is nonzero.
\begin{theorem}\label{nonzerorowintro}
    The randomized space complexity for \NonzeroRow~is $\wt{\Theta}(1)$, the pseudo-deterministic space complexity for \NonzeroRow~is $\wt{\Theta}(n)$, and the deterministic space complexity for \NonzeroRow~is $\wt{\Theta}(nd)$.
\end{theorem}

The idea behind the proof of Theorem \ref{nonzerorowintro} is to sample a random vector $x$, and then deterministically find a nonzero entry of $Ax$. With high probability, if a row of $A$ is nonzero, then the corresponding entry of $Ax$ will be nonzero as well.

\paragraph{Discussion:} Roughly speaking, in this problem there is a certain structure that allows us to use randomness to ``hash'' pieces of the input together, and then apply a deterministic algorithm on the hashed pieces. The other upper bounds we show for pseudo-deterministic algorithms also have a structure which allows us to hash, and then use a deterministic algorithm. It is interesting to ask if there are natural problems which have faster pseudo-deterministic algorithms than the best deterministic algorithms, but for which the pseudo-deterministic algorithms follow a different structure.

The next problems we show upper bounds for are estimating frequencies in a length-$m$ stream of elements from a large universe $[n]$ up to error $\eps m$, and that of estimating the inner product of two vectors $x$ and $y$ in an insertion-only stream of length-$m$ up to error $\eps\cdot\|x\|_1\cdot\|y\|_1$.  We show a separation between the deterministic and (weak) pseudo-deterministic space complexity in the regime where $m\ll n$.
\begin{theorem}
    There is a pseudo-deterministic algorithm for point query estimation and inner product estimation that uses $O\left(\frac{\log m}{\eps}+\log n\right)$ bits of space.  On the other hand, any deterministic algorithm needs $\Omega\left(\frac{\log n}{\eps}\right)$ bits of space.
\end{theorem}

\subsection{Related work}

Pseudo-deterministic algorithms were introduced by Gat and Goldwasser \cite{GG}. Such algorithms have since been studied in the context of standard (sequential algorithms) \cite{roots, OS}, average case algorithms \cite{dhiraj}, parallel algorithms \cite{matching}, decision tree algorithms \cite{GGR, k-pseudodeterminism}, interactive proofs \cite{proofs}, learning algorithms \cite{OS2}, approximation algorithms \cite{OS2, dixon}, and low space algorithms \cite{reproducibility}. In this work, we initiate the study of pseudo-determinism in the context of streaming algorithms (and in the context of one-way communication complexity).

The problem of finding duplicates in a stream of integers between 1 and $n$ was first considered by \cite{gopalan2009finding}, where an $O(\log^3 n)$ bits of space algorithm is given, later improved by \cite{jowhari2011tight} to $O(\log^2 n)$ bits. We show that in contrast to these low space randomized algorithms, a pseudo-deterministic algorithm needs significantly more space in the regime where the length of the stream is, say, $3n/2$. \cite{duplicates} shows optimal lower bounds for randomized algorithms solving the problem.

The method of $\ell_p$-sampling to sample an index of a turnstile vector with probability proportional to its $\ell_p$ mass, whose study was initiated in \cite{monemizadeh20101}, is one way of outputting an element from the support of a turnstile stream. A line of work \cite{frahling2008sampling,monemizadeh20101,jowhari2011tight,andoni2010streaming}, ultimately leading to an optimal algorithm in \cite{jayaram2018perfect} and tight lower bounds in \cite{duplicates}, characterizes the space complexity of randomized algorithms to output an element from the support of a turnstile vector as $\Theta(\poly\log n)$, in contrast with the space lower bounds we show for algorithms constrained to a low entropy output.

\subsection{Open Problems}

\paragraph{Morris Counters:} In \cite{morrisapxcounting}, Morris showed that one can approximate (up to a multiplicative error) the number of elements in a stream with up to $n$ elements using $O(\log \log n)$ bits of space. Does there exist an $O(\log \log n)$ bits of space pseudo-deterministic algorithm for the problem?

\paragraph{$\ell_2$-norm estimation:}  In this work, we show that there are no low-dimensional pseudo-deterministic sketching algorithms for estimating the $\ell_2$-norm of a vector.  However, we do not show a turnstile streaming lower bound for pseudo-deterministic algorithms, which motivates the following question.  Does there exist a $O(\poly\log n)$ space pseudo-deterministic algorithm for $\ell_2$-norm estimation?

\paragraph{Multi-pass streaming lower bounds:}  All the streaming lower bounds we prove are in the single pass model, i.e., where the algorithm receives the stream exactly once.  How do these lower bounds extend to the multi-pass model, where the algorithm receives the stream multiple times? All of the pseudo-deterministic streaming lower bounds in this paper do not even extend to 2-pass streaming algorithms.

\subsection{Table of complexities}
In the below table, we outline the known space complexity of various problems considered in our work.
\begin{table}[h]
\begin{tabular}{|l|l|l|l|}
\hline
\textbf{Problem}                   & \textbf{Randomized}               & \textbf{Deterministic}                   & \textbf{Pseudo-deterministic}                    \\ \hline
Morris Counters                    & $\Theta(\log\log n)$              & $\Theta(\log n)$                         & $O(\log n)$, $\Omega(\log \log n)$               \\ \hline
Find-Duplicate                     & $\Theta(\log n)$                  & $\Theta(n)$                              & $\widetilde{\Theta}(n)$                          \\ \hline
$\ell_2$-approximation (streaming) & \multirow{2}{*}{$\Theta(\log n)$} & \multirow{2}{*}{$\widetilde{\Theta}(n)$} & $\widetilde{\Theta}(n)$                          \\ \cline{1-1} \cline{4-4} 
$\ell_2$-approximation (sketching) &                                   &                                          & $\widetilde{O}(n)$, $\widetilde{\Omega}(\log n)$ \\ \hline
Find-Nonzero-Row                   & $\widetilde{\Theta}(1)$           & $\widetilde{\Theta}(nd)$                 & $\widetilde{\Theta}(n)$                          \\ \hline
\end{tabular}
\caption{\label{tab:space-tab}Table of space complexities.}
\end{table}

\section{Preliminaries}

A randomized algorithm is called \textit{pseudo-deterministic} if for every valid input $x$, when running the algorithm twice on $x$, the same output is obtained with probability at least $2/3$. Equivalently (up to amplification of error probabilities), one can think of an algorithm as pseudo-deterministic if for every input $x$, there is a unique value $f(x)$ such that with probability at least 2/3 the algorithm outputs  $f(x)$ on input $x$

\begin{definition}[Pseudo-deterministic]
A (randomized) algorithm $A$ is called \textit{pseudo-deterministic} if for all valid inputs $x$, the algorithm $A$ satisfies
\[\Pr_{r_1, r_2} [A(x, r_1) = A(x, r_2)] \ge 2/3.\]
\end{definition}

An extension of pseudo-determinism is that of $k$-entropy randomized algorithms \cite{reproducibility}. Such algorithms have the guarantee that for every input $x$, the distribution $A(x, r)$ (over a random choice of randomness $r$) has low entropy, in particular bounded by $k$.

Another extension of pseudo-determinism is that of $m$-pseudo-deterministic algorithms, from \cite{k-pseudodeterminism}. Intuitively speaking, any algorithm is $k$-pseudo-deterministic if for every valid input, with high probability the algorithm outputs one of $k$ options (so, a 1-pseudo-deterministic algorithm is the same as a standard pseudo-deterministic algorithm, since it outputs the one unique option with high probability):

\begin{definition}[$k$-pseudo-deterministic]
We say that an algorithm $A$ is \textit{$k$-pseudo-deterministic} if for all valid inputs $x$, there is a set $S(x)$ of size at most $k$, such that $\Pr_r [A(x,r) \in S(x)] \ge \frac{k+1}{k+2}$ 
\end{definition}

For the purposes of this work, we define a simple notion that we call a \emph{$k$-concentrated algorithm}.
\begin{definition}
    We say that an algorithm $A$ is \textit{$k$-concentrated} if for all valid inputs $x$, there is some output $F(x)$ such that $\Pr_r[A(x,r) = F(x)] \ge \frac{1}{k}$.
\end{definition}

The reason for making this definition is that any $\log k$-entropy randomized algorithm, and any $(k+2)$-pseudo-deterministic algorithm is $k$-concentrated. Thus, showing an impossibility result for $k$-concentrated algorithms also shows an impossibility result for $\log k$-entropy and $(k+2)$-pseudo-deterministic algorithms. Indeed, in this work, we use space lower bounds against $k$-concentrated algorithms to simultaneously conclude space lower bounds against low entropy and multi-pseudo-deterministic algorithms.

\begin{definition}
    A turnstile streaming algorithm is one where there is a vector $v$, and the input is a stream of updates of the form ``increase the $i^{\text{th}}$ coordinate of $v$ by $r$'' or ``decrease the $i^{\text{th}}$ coordinate of $v$ by $r'$''. The goal is to compute something about the final vector, after all of the updates.
\end{definition}






We use a pseudorandom generator for space-bounded computation due to Nisan \cite{nisan1992pseudorandom}, which we recap below.
\begin{theorem}
\label{thm:nisan-prg}
    There is a function $G:\{0,1\}^{s\log r}\rightarrow\{0,1\}^r$ such that
    \begin{enumerate}
        \item 
    Any bit of $G(x)$ for any input $x$ can be computed in $O(s\log r)$ space.
     \item For all functions $f$ from $\{0,1\}^r$ to some set $A$ such that $f$ is computable by a finite state machine on $2^s$ states, the total variation distance between the random variables $f(\bx)$ and $f(G(\by))$ where $\bx$ is uniformly drawn from $\{0,1\}^r$ and $\by$ is uniformly drawn from $\{0,1\}^{s\log r}$ is at most $2^{-s}$.
      \end{enumerate}
\end{theorem}

\section{\textsc{Find-Duplicate}: 
Pseudo-deterministic lower bounds}

Consider the following problem: the input is a stream of $3n/2$ integers between $1$ and $n$. The goal is to output a number $k$ which appears at least twice in the stream. Call this problem \FindRep. Recall that this problem has been considered in the past literature, specifically in \cite{gopalan2009finding,jowhari2011tight, duplicates}, where upper and lower bounds for randomized algorithms have been shown.

Indeed, we know the following is true from \cite{gopalan2009finding,jowhari2011tight}.
\begin{theorem} \label{thm:rand-alg-find-dup}
\FindRep~has an algorithm which uses $O(\poly\log n)$ memory and succeeds with all but probability $\frac{1}{\poly(n)}$.
\end{theorem}

We formally define a \emph{pseudo-deterministic streaming algorithm} and show a pseudo-deterministic lower bound for \FindRep~to contrast with the randomized algorithm from \pref{thm:rand-alg-find-dup}.

\begin{definition}[Pseudo-deterministic Streaming Algorithm]
A pseudo-deterministic streaming algorithm is a (randomized) streaming algorithm $A$ such that for all valid input streams $s = \langle x_1,\dots, x_m\rangle$, the algorithm $A$ satisfies $\Pr_{r_1, r_2}[(A(x, r_1) = A(x, r_2)] \ge 2/3$.
\end{definition}

One can also think of a pseudo-deterministic streaming algorithm as an algorithm $A$ such that for every valid input stream $s$, there exists some valid output $f(s)$ such that the algorithm $A$ outputs $f(s)$ with probability at least $2/3$ (one would have to amplify the success probability using repetition to see that this alternate notion is the same as the definition above).

\begin{definition}[\FindRep]
Define \FindRep~to be the streaming problem where the input is a stream of length $3n/2$ consisting of up to $n$, and the output must be an integer which has occured at least twice in the string.
\end{definition}

\begin{theorem} \label{thm:psd-lower-bound}
\FindRep~has no pseudo-deterministic algorithm with memory $o(n)$.
\end{theorem}
\paragraph{Proof Overview:}
In order to prove \pref{thm:psd-lower-bound}, we introduce two communication complexity problems --- \OWFindRep~and \OWPartRec:

In the \OWFindRep~problem,  Alice has a list of $3n/4$ integers between $1$ and $n$, and so does Bob. Alice sends a message to Bob, after which Bob must output an integer which is in both Alice's and Bob's list. Formally:

\begin{definition}[\OWFindRep]
Define \OWFindRep~to be the one-way communication complexity problem where Alice has input $S_A \subseteq [n]$ and Bob has input $S_B\subseteq[n]$, where $|S_A| , |S_B| \ge 3n/4$. The goal is for Bob to output an element in $S_A \cap S_B$.
\end{definition}

The idea is that one can reduce \OWFindRep~to \FindRep. So, our new goal will be to show that \OWFindRep~requires high communication. To do so, we will show that it is possible to reduce a different problem, denoted \OWPartRec (defined below), to \OWFindRep. Informally, in the \OWPartRec~problem, Alice has a list of $3n/4$ integers between $1$ and $n$. Bob does not have an input. Alice sends a message to Bob, after which Bob must output $n/10$ distinct elements which are all in Alice's list. Formally:

\begin{definition}[\OWPartRec]
Define \OWPartRec~to be the one-way communication complexity problem where Alice has input $S_A \subseteq [n]$ and Bob has no input. The goal is for Bob to output a set $S$ satisfying $S \subseteq S_A$ and $|S| \ge n/10$.
\end{definition}

We will show in Claim \ref{claim:find-rep-communication} that a low memory pseudo-deterministic algorithm for \FindRep~implies a low-communication pseudo-deterministic algorithm for \OWFindRep, and in Claim \ref{claim:find-rep-to-part-rec} that a low-communication pseudo-deterministic algorithm for \OWFindRep~implies a low communication algorithm for \OWPartRec. Finally, in Claim \ref{claim:part-rec-lower-bound}, we show that \OWPartRec~cannot be solved with low communication. Combining the claims yields Theorem \ref{thm:psd-lower-bound}.

\begin{proof}[Proof of Theorem \ref{thm:psd-lower-bound}]

\begin{claim}   \label{claim:find-rep-communication}
A pseudo-deterministic algorithm for \FindRep~with space $S$ and success probability $p$ implies a pseudo-deterministic communication protocol for \OWFindRep~with communication $S$ and success probability at least $p$.
\end{claim}

\begin{proof}
To prove the above claim, we construct a protocol for \OWFindRep~from a streaming algorithm for \FindRep. Given an instance of \OWFindRep, Alice can stream her input set of integers in increasing order, and simulate the streaming algorithm for \FindRep. Then, she sends the current state of the algorithm (which is at most $S$ bits) to Bob, who continues the execution of the streaming algorithm. At the end, the streaming algorithm outputs a repetition with probability $p$, which means the element showed up in both Alice and Bob's lists. Note that for a given input to Alice and Bob, Bob outputs a unique element with high probability because the streaming algorithm is pseudo-deterministic.
\end{proof}

\begin{claim}   \label{claim:find-rep-to-part-rec}
A pseudo-deterministic one-way communication protocol for \OWFindRep~with $S$ communication and failure probability $O\left(\frac{1}{n^2}\right)$ implies a pseudo-deterministic communication protocol for \OWPartRec~with $S$ communication and $O\left(\frac{1}{n}\right)$ failure probability.
\end{claim}

\begin{proof}
We will show how to use a protocol for \OWFindRep~to solve the instance of \OWPartRec.

Suppose we have an instance of \OWPartRec. Alice sends the same message to Bob as if the input was an instance of \OWFindRep, which is valid since in both of these problems, Alice's input is a list of length $3n/4$ of integers between $1$ and $n$.

Now, Bob's goal is to use the message sent by Alice to recover $n/10$ elements of Alice. Let $X$ be the (initially empty) set of elements of Alice's input that Bob knows and let $B$ be a set of $3n/4$ elements in $\{1, \dots, n\}$ disjoint from $X$, where we initially set $B$ to $\{1, 2, \dots, n\}$. While the size of $X$ is less than $n/10$, Bob simulates the protocol of \OWFindRep~with Alice's message and input $B$. This will result in Bob finding a single element $x$ in Alice's input that is (i) in $B$, and (ii) not in $X$. Bob adds $x$ to $X$, and deletes $x$ from $B$. Once the size of $X$ is $n/10$, Bob outputs $X$.

If Alice has the set $A$ as her input, define $f_A(B)$ to be the output which the pseudo-deterministic algorithm for \OWFindRep~outputs with high probability when Alice's input is $A$ and Bob's input is $B$. Now, set $B_0 = \{1, 2, \dots, n\}$, and $B_i = B_{i-1} \setminus \{f_A(B)\}$. Note that these $B_i$ (for $i = 0$ through $n/10$) are the sets which, assuming the pseudo-deterministic algorithm never errs during the reduction (where we say the algorithm errs if it does not output the unique element which is guaranteed to be output with high probability), Bob will use as his inputs for the simulated executions of \OWFindRep. The pseudo-deterministic algorithm does not err on any of the $B_i$ except with probability at most $1/n$, by a union bound. If Bob succeeds on all of the $B_i$, that means that the sequence of inputs which will be his inputs for the simulated executions of \OWFindRep~are indeed $B_0, B_1, \ldots, B_{n/10}$. So, since we have shown with high probability the algorithms succeeds on all of the $B_i$, and therefore with high probability the $B_i$ are also Bob's inputs for the simulated executions of \OWFindRep, we see that with high probability Bob will succeed on all of the $n/10$ inputs he tries to simulate executions of \OWFindRep~with.

Note that we used the union bound over all the $B_i$ for $i = 1$ through $n/10$. All of these $B_i$ are a function of $A$. In particular, notice that by definition, the $B_i$ do not depend on the randomness chosen by Alice.
\end{proof}

\begin{claim}   \label{claim:part-rec-lower-bound}
Every pseudo-deterministic \OWPartRec~protocol which succeeds with probability at least $\frac{2}{3}$ requires $\Omega(n)$ bits of communication.
\end{claim}

\begin{proof}
We prove this lower bound by showing that a protocol for \OWPartRec~can be used to obtain a protocol with exactly the same communication for the problem where Alice is given a string $x$ in $\{0,1\}^{Cn}$ as input, she sends a message to Bob, and Bob must exactly recover $x$ from Alice's message with probability at least 2/3. This problem has a lower bound of $\Omega(n)$ bits of communication.

Suppose there exists a pseudo-deterministic algorithm for \OWPartRec. Given such a pseudo-deterministic protocol that succeeds with probability at least $2/3$, there is a function $F$ such that $F(S)$ (a set with $n/10$ elements) is Bob's output after the protocol with probability at least $2/3$ when Alice is given $S$ as input.

We will construct sets $S_1,\dots,S_t$ to be subsets of $[n]$ of size $3n/4$ such that for any $i\ne j$, $F(S_i)$ is not a subset of $S_j$. To do so, we use the probabilistic method: set $S_1,\dots, S_t$ be random subsets of $[n]$ of size $3n/4$. The probability that $F(S_i)$ is contained $S_j$ for fixed $i \ne j$ is at most $\left(\frac{3}{4}\right)^{n/10}$. Thus, by a union bound, the probability that for any $i\ne j$, $F(S_i)$ is contained $S_j$ is at most $t^2\left(\frac{3}{4}\right)^{n/10}$, a quantity which is strictly less than $1$ when $t$ is $\left(\frac{4}{3}\right)^{n/100}$, so $S_1,\dots,S_t$ satisfying the desired guarantee exist.

Alice and Bob can (ahead of time) agree on an encoding of $\floor{\log t}$-bit strings that is an injective function $G$ from $\{0,1\}^{\floor{\log t}}$ to $\{S_1,\dots,S_t\}$. Now, if Alice is given a $\floor{\log t}$-bit string $x$ as input, she can send a message to Bob according to the pseudo-deterministic protocol for \OWPartRec~by treating her input as $G(x)$. Bob then recovers $F(G(x))$ with probability at least 2/3, and can use it to recover $G(x)$ since there is unique $S_i$ in which $F(G(x))$ is contained. Since $G$ is injective, Bob can also recover $x$ with probability $2/3$. 

This reduction establishes a lower bound of $\Omega(\floor{\log t})$ on the pseudo-deterministic communication complexity of \OWPartRec, which is an $\Omega(n)$ lower bound.
\end{proof}




Combining \pref{claim:find-rep-communication}, \pref{claim:find-rep-to-part-rec} and \pref{claim:part-rec-lower-bound} completes the proof of \pref{thm:psd-lower-bound}.
\end{proof}


It is worth noting that the problem has pseudo-deterministic algorithms with sublinear space if one allows multiple passes through the input. Informally, a $p$-pass streaming algorithm is a streaming algorithm which, instead of seeing the stream only once, gets to see the stream $p$ times. 

\begin{claim}
    There is a $p$-pass deterministic streaming algorithm that uses $\widetilde{O}(n^{1/p})$ memory for the \FindRep~problem.
\end{claim}
\begin{proof}
    At the start of $t$-th pass, the algorithm maintains a candidate interval $I$ of width $n^{1-(t-1)/p}$ from which it seeks to find a repeated element. At the very beginning, this candidate interval is $[1,n]$. In the $t$-th pass, first partition the interval into $n^{1/p}$ equal sized intervals $I'_1,\dots,I'_{n^{1/p}}$, each of whose width (the width of an interval $[a, b]$ is $b - a$) is $n^{1-t/p}$ and count the number of elements of the stream that lie in each such subinterval -- this count must exceed the width of at least one subinterval $I'_t$. Update $I$ to $I'_t$ and proceed to the next pass. After $p$ passes, this interval will contain at most 1 integer.
\end{proof}

\section{Entropy Lower Bound for \textsc{Find-Duplicate}}

\begin{theorem} \label{thm:nuanced-find-rep}
Every zero-error randomized algorithm for \FindRep~that is $\frac{n}{s}$-concentrated must use $\Omega\left(\frac{s}{\log n}\right)$ space.
\end{theorem}

By \textit{zero error}, we mean that the algorithm never outputs a number $k$ which is not repeated. With probability one it either outputs a valid output, or $\bot$.

\begin{proof}
We use a reduction similar to that of the pseudo-deterministic case (cf. Proof of \pref{claim:find-rep-to-part-rec}).  Using the exact same reduction from the proof of \pref{claim:find-rep-communication}, we get that a $\frac{n}{s}$-concentrated streaming algorithm for \FindRep~using $T$ space must give us a $\frac{n}{s}$-concentrated protocol for \OWFindRep~with communication complexity $T$. If we can give a way to convert such a protocol for \OWFindRep~into an $O\left(\frac{Tn\log n}{s}\right)$-communication protocol for \OWPartRec, the desired lower bound on $T$ follows from the lower bound on communication complexity of \OWPartRec~from \pref{claim:part-rec-lower-bound}. We will now show how to make such a conversion by describing a protocol for \OWPartRec.

Alice sends Bob $\Theta(n \log n/s)$ messages according to the protocol for \OWFindRep ~(that is, she simulates the protocol for \OWFindRep~ a total of $\Theta(n \log n/s)$ times). Bob's goal is to use these $\Theta(n \log n/s)$ messages to recover at least $n/10$ input elements of Alice. Towards this goal, he maintains a set of elements recovered so far, $X$ (initially empty), and a family of `active sets' $\calB$ (initially containing the set $\{1,2,\dots,n\}$). While the size of $X$ is smaller than $n/10$, Bob simulates the remainder of the \OWFindRep~protocol on every possible pair $(B, M)$ where $B$ is a set in $\calB$ and $M$ is one of the messages of Alice. For each such pair $(B,M)$ where the protocol is successful in finding a duplicate element $x$, Bob adds $x$ to $X$, removes $B$ from $\calB$ and adds $B\setminus\{x\}$ to $\calB$.

We now wish to prove that this protocol indeed lets Bob recover $n/10$ elements of Alice.
Suppose Alice has input $A$. For each set $S$, define $f_A(S)$ be an element of $A\cap S$ that has probability at least $s/n$ of being outputted by Bob on input $S$ at the end of a \OWFindRep~protocol.  Let $S_0 := \{1,2,\dots,n\}$ and $S_i :=  S_{i-1}\setminus\{f_A(S_i)\}$ be defined for $0\le i\le n/10$.  Note that $S_i$ are predetermined: it is a function of Alice's input (and, in particular, not a function of the randomness she uses when choosing her messages). For a fixed $i$, the probability of failure to recover $f_A(S_i)$ from any of Alice's messages is at most $1/n^2$. A failure to fill in $X$ with $n/10$ elements implies that for some $i$, Bob failed to recover $f_A(S_i)$ from all of Alice's messages. The probability that such a failure happens for a specific $i$ is at most $(1 - s/n)^{\Theta(n \log n / s)}$. By setting the constant in the $\Theta$ to be large enough, we can have this be at most $\frac{1}{n^2}$, and so by a union bound the probability that there is an $i$ such that $f_A(S_i)$ is not recovered by Bob is at most $1/n$.

Thus, we obtain a protocol for \textsc{One-Way-Partial-Recovery} with communication complexity $O(T n \log n / s)$, and so $T \le s/\log n$, completing the proof.
\end{proof}

We obtain the following as immediate corollaries:
\begin{corollary}
    Any zero-error $\log\left(\frac{n}{s}\right)$-entropy randomized algorithm for \FindRep~must use $\Omega\left(\frac{s}{\log n}\right)$ space.
\end{corollary}

\begin{corollary}
    Any zero-error $O\left(\frac{n}{s}\right)$-pseudo-deterministic algorithm for \FindRep~must use $\Omega\left(\frac{s}{\log n}\right)$ space.
\end{corollary}

Below we show that the above lower bound is tight up to log factors.

\begin{theorem} \label{thm:find-rep-low-entropy-algo}
For all $s$, there exists a zero-error randomized algorithm for \FindRep~using $\wt{O}(s)$ space (where $\wt{O}$ hides factors polylogarithmic in $n$) that is $O\left(\frac{n}{s}\right)$-concentrated.
\end{theorem}

\begin{proof}
Define the following algorithm $A$ for \FindRep: pick a random number $i$ in $[3n/2]$, then remember the $i^{\text{th}}$ element $a$ of the stream, and see if $a$ appears again later in the stream. If it does, return $x$. Otherwise return $\bot$. 

The $O\left(\frac{n}{s}\right)$-concentrated algorithm algorithm is as follows: Run $s\log n$ copies of Algorithm $A$ independently (in parallel), and then output the minimum of the outputs.

We are left to show that this algorithm is indeed $O\left(\frac{n}{s}\right)$-concentrated.

Define $f$ to be a function where $f(i)$ is the total number of times which $i$ shows up in the stream, and define $g(i) = \max((f(i)-1), 0)$. Note that then, the probability that $i$ is outputted by algorithm $A$ is $g(i)/(3n/2)$, since $i$ will be outputted if $A$ chooses to remember one of the first $i - 1$.

Consider the smallest $a$ such that $\sum_{i=1}^{a} g(i) \ge n/(2s)$. We will show that the probability that the output is less than $a$ with high probability. It will follow that the algorithm is $s$-concentrated, since of the $a-1$ smallest elements, at most $\sum_{i=1}^{a-1} g(i)$ outputs are possible (since if $g(i) = 0$, then $i$ is not a possible output). So, we will see that with high probability, one of at most $\sum_{i=1}^{a-1} g(i) + 1 \le n/(2s)$ outputs (namely, the valid outputs less than or equal to $a$) will be outputted with high probability. And hence, at least one of them will be outputted with probability at least $\frac{s}{n}$.

The probability that the output is at most $a$ in a single run of algorithm $A$ is $\frac{3n}{2}\sum_{i=1}^{a} g(i) \ge 3/(4s)$. So, the probability that in $s \log n$ runs of algorithm, in at least one of them an element which is at most $a$ is outputted is $1 - (1 - \frac{3}{4s})^{s \log n}$, which is polynomially small in $n$. Hence, with high probability an element which is at most $a$ (and there are $n/(2s)$ valid outputs less than $a$) will be outputted.


\end{proof}

\subsection{Getting Rid of the Zero Error Requirement}\label{nonzeroerrorsection}

A downside of \pref{thm:nuanced-find-rep} is that it shows a lower bound only for zero-error algorithms. In this section, we strengthen the theorem by getting rid of that requirement:

\begin{theorem} \label{thm:nuanced-find-rep-no-zero}
Every randomized algorithm for \FindRep~that is $\frac{n}{s}$-concentrated and errs with probability at most $\frac{1}{n^2}$ must use $\tilde{\Omega}\left(s^{1-\epsilon}\right)$ space (for all $\epsilon > 0$).
\end{theorem}

\paragraph{Proof overview:}
We begin by outlining why the approach of \pref{thm:nuanced-find-rep} does not work without the zero-error requirement. Recall that the idea in the proof was to have Alice send many messages (for \OWFindRep) to Bob, and Bob simulates the \OWFindRep~ algorithm (using simulated inputs he creates for himself) using these messages to find elements in Alice's input.

The problem is that the elements we end up removing from Bob's simulated input\footnote{recall that Bob simulates an input to the \OWFindRep~problem, and then he repeatedly finds elements he shares with Alice, removes them from the ``fake'' input, and reconstructs a large fraction Alice's inputs} depend on Alice's messages, and therefore we can't use a union bound to bound the probability that the protocol failed for a certain simulated input.  So, we want the elements we remove from Bob's fake input not to depend on the inputs Alice sent. One idea to achieve this is to have Alice send a bunch of messages (for finding a shared element), and then Bob will remove the element that gets output the largest number of times (by simulating the protocol with each of the many messages Alice sent). The issue with this is that if the two most common outputs have very similar probability, the outputted element depends not only on Alice's input, but also on the randomness she uses when choosing what messages to send to Bob. This makes it again not possible to use a union bound.

There are two new ideas to fix this issue. The first is to use the following ``Threshold'' technique: Bob will pick a random ``threshold'' T between $ks/(2n)$ and $ks/(4n)$ (where we wish to show a lower bound on $n/s$-concentrated algorithms, and $k$ is the total number of messages Alice sends to Bob). He simulates the algorithm for \OWFindRep~with all $k$ messages Alice sent him, and gets a list $L$ of $k$ outputs. Then, he will consider the ``real'' output to be the lexicographically first output $y \in L$ where there are more than $T$ copies of $y$ in the list $L$ (note that since the algorithm is $n/s$-concentrated, its very unlikely for no such element to exist).

Now, it follows that with high probability, the shared element does not really depend on the messages. This is because with all but probability approximately $1/\sqrt{ks/n}$, the threshold is far (more than $\sqrt{ks/n} \log^2 ks/n$ away) from the the frequency of every element in $L$. We note that we pick $\sqrt{ks/n} \log^2 ks/n$ since from noise we would expect to have the frequencies of elements in $L$ change by up to $\sqrt{ks/n} \log^2 ks/n$, depending on the randomness of $A$. We want the threshold to be further than that from the expected frequencies, so that with high probability there will be no element which sometimes has frequency more than $T$ and sometimes has frequency less than $T$, depending on Alice's messages (recall that the goal is to make the outputs depend as little as possible on Alice's messages, but to only depend on shared randomness and on Alice's input).

This is still not enough for us: we still cannot use a union bound, as $1/\sqrt{ks/n}$ fraction of the time Bob's output will depends on Alice's message (and not just her input). The next idea resolves this. What Alice will do is send $n/\sqrt{ks/n}$ additional pieces of information: telling Bob where the chosen thresholds are bad, and what threshold to use instead. We assume that we have shared randomness so Alice knows all of the thresholds that will be chosen by Bob (note heavy-recovery is hard, even in the presence of shared randomness, so the lower bound is sufficient with shared randomness). Now, Alice can tell for which executions there the threshold chosen will be too close to the likelihood of an element. So, Alice will send approximately $n/\sqrt{ks/n}$ additional pieces of information: telling Bob where the chosen thresholds are bad, and what threshold to use instead. By doing so, Alice has guaranteed that a path independent of her messages will be taken.

To recap, idea 1 is to use the threshold technique so that with probability $1-1/\sqrt{ks/n}$ what Bob does doesn't depend on Alice's messages (only on her input). Idea 2 is to have Alice tell Bob where these $1/\sqrt{ks/n}$ bad situations are, and how to fix them.

The total amount of information Alice sends (ignoring logs) is $\tilde{\Theta}(kb+n/\sqrt{ks/n})$, (where $b$ is the message size we are assuming exists for pseudo-deterministically finding a shared element, and k is the number of messages Alice sends). The factor $n/\sqrt{ks/n}$ follows since $1/\sqrt{ks/n}$ of the times, short messages will be sent to Bob due to a different threshold. A threshold requires $\log n$ bits to describe, which can be dropped since we are ignoring log factors. Setting $n/s \ll k \ll n/b$, we conclude that Alice sends a total of $\tilde{o}(n)$ bits. This establishes a contradiction, since we need $\tilde{\Theta}(n)$ bits to solve \OWPartRec. So, whenever $s = \tilde{\omega}(b)$, we can pick a $k$ such that we get a contradiction.

\begin{proof}
Below we write the full reduction written out as an algorithm for \OWFindRep.

\begin{itemize}
\item
Alice Creates $k=n/\sqrt{sb}$ messages for \OWFindRep, and sends them to Bob (Call these messages of type A).

\item 
Additionally, Alice looks at the thresholds in the shared randomness. every time there is a threshold that is close (within $\sqrt{ks/n} \log^2 ks/n$) of the expected number of times a certain $y$ will be outputted on the corresponding input (that is, for each fake input Bob will try, Alice checks if the probability of outputting some $y$ is close to $T$ -- to be precise, say she checks if its probability of being outputted, assuming a randomly chosen message by Alice, is close to $T$), she sends a message to Bob informing him about the bad threshold, and suggests a good threshold to be used instead (call these messages of type B). Notice that these messages do not depend on the messages of type A that Alice sends, and that each such message is of size $O(\log n)$.

\item
Bob sets $B$ to be the simulated input $\{1,...,n\}$

\item
Bob uses each of the messages of type $A$ that Alice sent, along with $B$, to construct a list of outputs.

\item
Bob looks at the shared randomness to find a threshold $T$ (if Alice has informed him it is a bad threshold, use the threshold Alice suggests instead), and consider the lexicographically minimal output y that is contained in the multiset more than $T$ times.

\item
Bob removes $y$ from the fake input and repeat the last three steps of the algorithm (this time using a new threshold).
\end{itemize}

\begin{claim}
The above protocol solves \OWPartRec with high probability using $\tilde{o}(n)$ bits.
\end{claim}

\begin{proof}
First we show that the total number of bits communicated is $\tilde{o}(n)$. Notice that the total number of messages of type $A$ that are sent is $n/\sqrt{sb}$. We assume that each of these is of size at most $b$, giving us a total of $n\sqrt{b}/\sqrt{s}$ bits sent in messages of type $A$. Under the assumption that $b = \tilde{o}(n)$, we see that this is $\tilde{o}(n)$ total bits for messages of type $A$. 

We now count the total number of bits communicated in messages of type $B$. Each message of type $B$ is of size $O(\log n)$ (it is describing a single element, and a number corresponding to which execution the message is relevant for, each requiring $O(\log n)$ bits). So, we wish to show that with high probability the total number of messages of type $B$ is $\tilde{o}(n)$. The total number of messages of type $B$ that will be sent is $O(\frac{n}{\sqrt{ks/n}})$, since for every input, the probability that the randomly chosen threshold (which is sampled using public randomness) is more than $\sqrt{ks/n} \log^2 ks/n$ away from the frequency of every output is $O(\frac{1}{\sqrt{ks/n}})$. Note that $\frac{n}{\sqrt{ks/n}} =  \tilde{o}(n)$ since $ks = n\sqrt{\frac{s}{b}}$, and we assume $b = \tilde{o}(s)$.

We are now left to show the protocol correctly solves $\OWPartRec$ with high probability. We will first show that, after fixing Alice's input and the public randomness, with high probability there will be a single sequence of inputs that Bob will try that will occur with high probability (that is, there is a sequence of $y$'s that Bob goes through with high probability). To do this, consider a certain input that Bob tries. We will bound the probability that there are two values $y$ and $y'$ such that both $y$ and $y'$ have probability at least $\frac{1}{n}$ of being outputted. Suppose there exists two such $y$ and $y'$ that means that at least one of them (say $y$, without loss of generality) has to be the output of more than $T$ of the $k$ executions with probability more than $\frac{1}{n}$, but less than $\frac{n-1}{n}$. Additionally, we know that the expected number of times that $y$ will be outputted of the $k$ times is more than $\sqrt{ks/n}\log^2 ks/n$ away from $T$ (otherwise Alice will pick a different value of $T$ such that this will be true, and send that value to Bob in a message of type $B$). However, the probability of being more than $\sqrt{ks/n}\log^2 ks/n = \Theta(\left(\frac{s}{b}\right)^{1/4}\log^2 s/b) = \Theta(n^{\eps/4} \log^2 n)$ away from the expectation, by a Chernoff bound, is (asymptotically) less than $\frac{1}{n}$.

Notice also, that by the assumption that the algorithm in $n/s$-concentrated, there will always be an output $y_{\max}$ which is expected to appear at least $\frac{s}{n}$ of the time. Also, since the threshold $T$ is at most $ks/2n$, the probability that $y_{\max}$ appeared fewer than $T$ times is exponentially low in $ks/n = \sqrt{s/b} = \tilde{\Theta}(n^{\epsilon/2})$, and so with high probability there will always exist a $y$ which was outputted on more than $T$ of the executions, so in the second to last step, the multiset will always have an element that appears at least $T_1$ times.

Hence, by a union bound over all inputs that Bob tries, with high probability there will be a single sequence of inputs which Bob goes through (which depends only only on the public thresholds and Alice's input).

We will show that each $y$ generated by Bob is an element in Alice's input with high probability. Notice that the $y$ that Bob picks has appeared more than $T$ times out of $k$, where $T$ is at least $ks/(4n)$. If $y$ is not a valid output then its probability of being outputted is $\frac{1}{n^2}$. The probability it is outputted at least once is at most $\frac{k}{n^2} \le \frac{1}{n}$. Taking a union bound over the inputs that Bob tries (of which there are $n/10$), we get that the probability that there is an invalid $y$ at any point is at most 1/10. So, with probability 9/10, no invalid $y$ is ever outputted.
\end{proof}

\end{proof}

\section{Entropy lower bounds for finding a support element}
Consider the turnstile model of streaming, where a vector $z\in\R^n$ starts out as $0$ and receives updates of the form `increment $z_i$ by 1' or `decrement $z_i$ by 1', and the goal of outputting a nonzero coordinate of $z$. This is a well studied problem and a common randomized algorithm to solve this problem in a small amount of space is known as $\ell_0$ sampling \cite{frahling2008sampling}.  $\ell_0$ sampling uses polylogarithmic space and outputs a uniformly random coordinate from the support of $z$. A natural question one could ask is whether the output of any low space randomized algorithm is necessarily close to uniform, i.e., has high entropy. We answer this affirmatively and show a nearly tight tradeoff between the space needed to solve this problem and the entropy of the output of a randomized algorithm under the assumption that the algorithm is not allowed to output anything outside the support\footnote{We note that using similar ideas to those in Subsection \ref{nonzeroerrorsection}, the zero error requirement could be removed. We omit this adaptation since it is very similar to that of Subsection \ref{nonzeroerrorsection}.}

\begin{theorem}
    Every zero-error randomized algorithm for \FindSupElem~that is $\frac{n}{s}$-concentrated must use $\Omega\left(\frac{s}{\log n}\right)$ space.
\end{theorem}
We only provide a sketch of the proof and omit details since they are nearly identical to the proof of \pref{thm:nuanced-find-rep}.

\begin{proof}[Proof Sketch]
    Let $\mathcal{A}$ be such an algorithm that uses $T$ space. Just like the proof of \pref{thm:nuanced-find-rep}, the way we show this lower bound is by illustrating that $\mathcal{A}$ can be used to obtain an $O\left(\frac{Tn\log n}{s}\right)$-communication protocol for \OWPartRec, which combined with \pref{claim:part-rec-lower-bound} yields the desired result.
    
    For every element $a$ in Alice's input set $A$, she streams `increment $z_a$ by 1' and runs $\Theta\left(\frac{n}{s}\log n\right)$ independent copies of $\mathcal{A}$ on the input. She then sends the states of each these independent runs of $\mathcal{A}$ to Bob, which is at most $\frac{Tn\log n}{s}$ bits, to Bob. Bob maintains a set of states $\mathcal{M}$, initially filled with all of Alice's messages. While he has not yet recovered $n/10$ elements, Bob picks a message $M\in\mathcal{M}$ and recovers $x$ in $A$ using algorithm $\mathcal{A}$. And for each $M\in\mathcal{M}$, Bob resumes $\mathcal{A}$ on state $M$ and streams `decrement $z_x$ by 1' and adds the new state to $\mathcal{M}$, and deletes $M$ from $\mathcal{M}$.
    
    The proof of correctness for why Bob indeed eventually recovers $n/10$ elements of $A$ is identical to that in the proof of \pref{thm:nuanced-find-rep}, thus giving a protocol for \OWPartRec~and proving the statement.
\end{proof}

We can immediately conclude the following.
\begin{corollary}
    Any zero-error $\log\left(\frac{n}{s}\right)$-entropy randomized algorithm for \FindSupElem~must use $\Omega\left(\frac{s}{\log n}\right)$ space.
\end{corollary}

\begin{corollary}\label{cor:multi-psd-findsupelem}
    Any zero-error $O\left(\frac{n}{s}\right)$-pseudo-deterministic algorithm for \FindSupElem~must use $\Omega\left(\frac{s}{\log n}\right)$ space.
\end{corollary}

This lower bound is also tight up to polylogarithmic factors due to an algorithm nearly identical to the one from \pref{thm:find-rep-low-entropy-algo}. In particular, we have:
\begin{theorem}
For all $s$, there exists a zero-error randomized algorithm for \FindRep~using $O(s)$ space that is $O\left(\frac{n}{s}\right)$-concentrated.
\end{theorem}

\section{Space complexity of pseudo-deterministic $\ell_2$-norm estimation}
In this section, we once again consider the pseudo-deterministic complexity of $\ell_2$ norm estimation in the \emph{sketching model}.  The algorithmic question here is to design a distribution $\calD$ over $s\times n$ matrices along with a function $f:\mathbb{R}^s\rightarrow\mathbb{R}$ so that for any $x\in\mathbb{R}^n$:
\[
    \Pr_{\bS\sim\calD}[f(\bS x)\not\in\left[\frac{1}{\alpha}\|x\|_2,\alpha\|x\|_2\right] \le \frac{1}{\poly(n)}.
\]
Further, we want $f(\bS x)$ to be a pseudo-deterministic function; i.e., we want $f(\bS x)$ to be a unique number with high probability.

\begin{theorem}\label{thm:l2-main}
    The pseudo-deterministic sketching complexity of $\ell_2$ norm estimation is $\Omega(n)$.
\end{theorem}




The following query problem is key to our lower bound.
\begin{definition}[$\ell_2$ adaptive attack]\label{def:adaptive-attack}
    Let $\alpha > 0$ be some constant. Let $S$ be an  $s\times n$ matrix with real-valued entries and $f:\R^s \rightarrow \R$ be some function.  Now, consider the query model where an algorithm is allowed to specify a vector $x\in\R^n$ as a query and is given $f(Sx)$ as a response.  The goal of the algorithm is to output $y$ such that
    \[
        f(Sy)\notin\left[\frac{1}{\alpha}\|y\|_2,\alpha\|y\|_2\right]
    \]
    in as few queries as possible.  We call this algorithmic problem the \emph{$\ell_2$-adaptive attack problem}.
\end{definition}

We use a theorem on adaptive attacks on $\ell_2$ sketches proved in \cite{hardt2013robust}.
\begin{theorem}\label{thm:adaptive-attack}
    There is a $\poly(n)$-query protocol to solve the $\ell_2$ adaptive attack problem with probability at least $9/10$, i.e., the problem in \pref{def:adaptive-attack} when $s=o(n)$.
\end{theorem}

\begin{proof}[Proof of \pref{thm:l2-main}]
Suppose $\calD$ is a distribution over $s\times n$ sketching matrices and $f$ is a function mapping $\R^s$ to $\R$ with the property that the pair $(\calD,f)$ gives a pseudo-deterministic sketching algorithm for $\ell_2$ norm estimation.  Henceforth, we use $\bS$ to denote a random matrix sampled from $\calD$.  Then there is a function $g:\R^n\rightarrow \R$ such that:
\begin{enumerate}
    \item $g$ is an $\alpha$-approximation of the $\ell_2$ norm.
    \item On every input $x$, $f(\bS x)=g(x)$ with probability at least $1-\frac{1}{n^c}$ for some constant $c$.
\end{enumerate}
We will show that $s$ must be $\Omega(n)$ by deducing a contradiction when $k = o(n)$.  Let $r$ be a parameter to be chosen later.  Let $\bx^{(1)},\bx^{(2)},\dots,\bx^{(r)}$ be the (random) sequence of vectors in $\R^n$ obtained by the adaptive query protocol from \pref{thm:adaptive-attack} based on responses $g(\bx^{(0)}),\dots,g(\bx^{(r)})$ where $r=\poly(n)$, and let $\by$ be the (random) output of the protocol.  Note that the guarantee that $r=\poly(n)$ hinges on assuming $s=o(n)$.  From the guarantees of \pref{thm:adaptive-attack}, for any fixed matrix $B$ and function $h$ such that $h(B\bx^{(i)}) = g(\bx^{(i)})$ for all $i$, it is true with probability at least $9/10$ that $h(B\by) \ne g(\by)$.  On the other hand, for any sequence of $r+2$ fixed vectors $v_0,\dots,v_{r+1}$, $f(\bS v_i) = g(v_i)$ for all $i$ with probability at least $1-\frac{1}{\poly(n)}$.  Call the event $\{f(\bS\bx^{(0)})=g(\bx^{(0)}),\dots,f(\bS\bx^{(r)})=g(\bx^{(r)}),f(\bS\by)=g(\by)\}$ as $\mathcal{E}$.  Let $p_S$ be the probability density function of $\bS$ and let $p_T$ be the probability density function of $(\bx^{(0)},\dots,\bx^{(r)},\by)$.  This results in the following two estimates of $\Pr[\mathcal{E}]$.

On the one hand,
\begin{align*}
    \Pr[\mathcal{E}] &= \int_{\bS}\Pr[\mathcal{E}|\bS]p_S(\bS)\\
    &\le \int_{\bS}\frac{1}{10}p_S(\bS)\\
    &= \frac{1}{10},
\end{align*}
and on the other hand,
\begin{align*}
    \Pr[\mathcal{E}] &= \int_{\bx^{(0)},\dots,\bx^{(r)},\by}\Pr[\mathcal{E}|\bx^{(0)},\dots,\bx^{(r)},\by]p_T(\bx^{(0)},\dots,\bx^{(r)},\by)\\
    &\ge \int_{\bx^{(0)},\dots,\bx^{(r)},\by}\left(1-\frac{1}{\poly(n)}\right)p_T(\bx^{(0)},\dots,\bx^{(r)},\by)\\
    &= 1-\frac{1}{\poly(n)}.
\end{align*}

The contradiction arises since $\Pr[\mathcal{E}]$ cannot simultaneously be at least $1-\frac{1}{\poly(n)}$ and at most $\frac{1}{10}$, and hence $s$ cannot be $o(n)$.

\end{proof}

\begin{corollary}   \label{cor:tight-ell_2-lower}
    For any constant $\delta > 0$, any $(2-\delta)$-concentrated sketching algorithm that where the sketching matrix is $s\times n$ can be turned into a pseudo-deterministic one by running $\log n$ independent copies of the sketch and outputting the majority answer.  Thus, as an upshot of \pref{thm:l2-main} we obtain a lower bound of $\Omega\left(\frac{n}{\log n}\right)$ on $(2-\delta)$-concentrated algorithms for pseudo-deterministic $\ell_2$-norm estimation in the sketching model.
\end{corollary}

In contrast to \pref{cor:tight-ell_2-lower} which says that $(2-\delta)$-concentrated algorithms for $\ell_2$ estimation in the sketching model need near linear dimension, we show that there is an $O(\poly\log n)$-dimension $(2+\delta)$-concentrated sketching algorithm to solve the problem, thus exhibiting a `phase transition'.
\begin{theorem} \label{thm:low-conc-alg-l2}
    There is a distribution $\calD$ over $s\times n$ matrices and a function $f:\R^s\rightarrow\R$ when $s=O(\poly\log n)$
    For every constant $\delta > 0$, there is an $O(\poly(\log n,\log m))$-space $(2+\delta)$-concentrated sketching algorithm for $\ell_2$-norm estimation. 
\end{theorem}
\begin{proof}
    Let the true $\ell_2$ norm of the input vector be $r$.  Run the classic sketching algorithm of \cite{alon1999space} for randomized $\ell_2$ norm estimation with error $\min\{1/2^{20},\eps^4\}$ and failure probability $\frac{1}{\poly(n)}$ where $(1+\eps)$ is the desired approximation ratio.  This uses a sketch of dimension $O(\poly\log n)$.  Now, we describe the function $f$ we use.  Take the output of the sketching algorithm of \cite{alon1999space} and return the number obtained by zeroing out all its bits beyond the first $\max\{2\log\left(\frac{1}{\eps}\right),5\}$ significant bits.\footnote{The parameters $1/2^{20},5$ and $\eps^4$ are chosen purely for \emph{safety} reasons}  First, the outputted number is a $(1+\eps)$ approximation.  Further, for each input, the output is one of two candidates with probability $1-\frac{1}{\poly(n)} > 1-\delta$ for every constant $\delta$.  This is because \cite{alon1999space} produces a $(1+\eps^4)$-approximation to $r$, and there are only two candidates for the $2\log\left(\frac{1}{\eps}\right)$ most significant bits of any real number that lies in an interval $[(1-\eps^4)r,(1+\eps^4)r]$.
\end{proof}

\section{Pseudo-deterministic Upper Bounds}




\subsection{Finding a nonzero row}
Given updates to an $n\times d$ matrix $A$ (where we assume $d\le n$) that is initially $0$ in a turnstile stream such that all entries of $A$ are always in range $[-n^3,n^3]$, the problem \NonzeroRow~is to either output an index $i$ such that the $i$th row of $A$ is nonzero, or output \texttt{none} if $A$ is the zero matrix. 

\begin{theorem}
The randomized space complexity for \NonzeroRow~is $\wt{\Theta}(1)$, the pseudo-deterministic space complexity for \NonzeroRow~is $\wt{\Theta}(n)$, and the deterministic space complexity for \NonzeroRow~is $\wt{\Theta}(nd)$.
\end{theorem}

\begin{proof}

We first will show a randomized $\wt{\Theta}(1)$ space algorithm for the problem, then we will show pseudo-deterministic upper and lower bounds, and then show the deterministic lower bound.

\paragraph{Randomized algorithm for \NonzeroRow.}
A randomized algorithm for this problem is given below.  Note that the version of the algorithm as stated below does not have the desirable $\wt{O}(1)$ space guarantee, but we will show how to use a pseudorandom generator of Nisan \cite{nisan1992pseudorandom} to convert the below algorithm to one that uses low space.
\begin{enumerate}
    \item \label{step:sample-vec} Sample a random $d$-dimensional vector $\bx$ where each entry is an independently drawn integer in $[-n^3, n^3]$ and store it.
    \item \label{step:update-y} Simulate a turnstile stream which maintains $A\bx$.  In particular, consider the $n$-dimensional vector $y$, which is initially $0$, and for each update to $A$ of the form ``add $\Delta$ to $A_{ij}$'', add $\Delta\bx_j$ to $y_i$.  We run an $\ell_0$-sampling algorithm \cite{frahling2008sampling} on this simulated stream updating $y$, and return the output of the $\ell_0$-sampler, which is close in total variation distance to a uniformly random element in the support of $y$.
\end{enumerate}
In the above algorithm, \pref{step:sample-vec} is not low-space as stated.  Before we give a way to perform \pref{step:sample-vec} in $\wt{O}(1)$ space, we prove the correctness of the above algorithm.  Suppose $A_i$ is a nonzero row of $A$, then let $j$ be an index where $A_i$ is nonzero.  Suppose all coordinates of $\bx$ except for the $j$-th coordinate have been sampled, there is at most one value $C$ for $\bx_j$ for which $\langle A_i, \bx\rangle$ is $0$, and there is at most a $1/n^3$ probability that $\bx_j$ equals $C$, which means if $i$ is a nonzero row, then $(A\bx)_i$ is nonzero except with probability at most $1/n^3$.  In fact, by taking a union bound over all nonzero rows we can conclude that the set of nonzero rows and the set of nonzero indices of $A\bx$ are exactly the same, except with probability bounded by $1/n^2$.

Now we turn our attention to implementing \pref{step:sample-vec} in low space.  Towards doing so we use Nisan's pseudorandom generator for space bounded computation in a very similar manner to \cite{indyk2006stable}.

Instead of sampling $3d\log n+1$ bits to store $\bx$, we sample and store a uniformly random seed $\bw$ of length $O(\poly\log(n,d))$ and add $\Delta G(\bw)_j$ to $y_i$ when an update ``add $\Delta$ to $A_{ij}$'' is received, where $G$ is the function from \pref{thm:nisan-prg} that maps the random seed to a sequence $3d\log n+1$ bits.  To prove the algorithm is still correct if we use the pseudorandom vector $G(\bw)$ instead of the uniformly random vector $\bx$, we must show that when $A_i$ is nonzero, then $\langle A_i,G(\bw)\rangle$ is nonzero with probability at least $1-O(1/n^3)$.  Towards this, for a fixed $d$-dimensional vector $q$, consider the following finite state machine.  The states are labeled by pairs $(i,a)$ where $i$ is in $\{0,1,\dots,d\}$ and $a$ is in $[-n^6d,n^6d]$.  The FSM takes a $d$-dimensional vector $r$ as input, starts at state $(0,0)$, and transitions from state $(i,a)$ to $(i+1,a+q_{i+1}\cdot r_{i+1})$ until it reaches a state $(d,\ell)$.  The FSM then outputs $\ell$.  This establishes that for a fixed $q$, the function $f(x)\coloneqq\langle q, x\rangle$ is computable by an FSM on $\poly(d,n)$ states, and hence from \pref{thm:nisan-prg}, $f(\bx)$ and $f(G(\bw))$ are $1/dn^6$ close in total variation distance, which means when $A_i$ is nonzero, then $\langle A_i,G(\bw)\rangle$ is nonzero except with probability bounded by $O(1/n^3)$.


\paragraph{A pseudo-deterministic algorithm and lower bound for \NonzeroRow}  The pseudo-deterministic algorithm is very similar to the randomized algorithm from the previous section.
\begin{enumerate}
    \item \label{step:sample-vec} Sample a random $d$-dimensional vector $\bx$ where each entry is an independently drawn integer in $[-n^3, n^3]$.  Store $\bx$ and maintain $A\bx$.
    \item Output the smallest index $i$ such that $(A\bx)_i$ is nonzero.
\end{enumerate}
Storing $\bx$ takes $O(d\log n)$ space, and maintaining $A\bx$ takes $O(n\log n)$ space.  Recall from the discussion surrounding the randomized algorithm that the set of nonzero indices of $A\bx$ and the set of nonzero rows were equal with probability $1-1/n^2$, which establishes correctness of the above pseudo-deterministic algorithm.  The space complexity is thus $O((d+n)\log n)$, which is equal to $O(n\log n)$ from the assumption that $d\le n$.

A pseudo-deterministic lower bound of $\wt{\Omega}(n)$ follows immediately from \pref{cor:multi-psd-findsupelem} since \NonzeroRow~specialized to the $d=1$ case is the same as \FindSupElem.

\paragraph{Lower Bound for deterministic algorithms.}  
An $\Omega(nd\log n)$ bit space lower bound for deterministic algorithms follows from a reduction to the communication complexity problem of \Equality.  Alice and Bob are each given $nd\log n$ bit strings as input, which they interpret as $n\times d$ matrices, $A$ and $B$ respectively, where each entry is a chunk of length $\log n$.  Suppose a deterministic algorithm $\calA$ takes $S$ bits of space to solve this problem.  We will show that this can be converted to a $S$-bit communication protocol to solve \Equality.  Alice runs $\calA$ on a turnstile stream updating matrix $X$ initialized at 0 by adding $A_{ij}$ to $X_{ij}$ for all $(i,j)$ in $[n]\times[d]$.  Alice then sends the $S$ bits corresponding to the state of the algorithm to Bob and he continues running $\calA$ on the updates `add $-B_{ij}$ to $X_{ij}$'.  $\calA$ outputs \texttt{none} if and only if $A = B$ and thus Bob outputs the answer to \Equality~depending on the output of $\calA$.  Due to a communication complexity lower bound of $\Omega(nd\log n)$ on \Equality, $S$ must be $\Omega(nd\log n)$.

\end{proof}

\subsection{Point Query Estimation and Inner Product Estimation} \label{sec:point-query-inner}
In this section, we give pseudo-deterministic algorithms that beat the deterministic lower bounds for two closely related streaming problems --- point query estimation and inner product estimation.
\paragraph{Point Query Estimation.}  Given a parameter $\eps$ and a stream of $m$ elements where each element comes from a universe $[n]$, followed by a query $i\in [n]$, output $f_i'$ such that $|f_i-f_i'|\le\eps m$ where $f_i$ is the frequency of element $i$ in the stream.

\paragraph{Inner Product Estimation.}  Given a parameter $\eps$ and a stream of $m$ updates to (initially $0$-valued) $n$-dimensional vectors $x$ and $y$ in an insertion-only stream\footnote{A stream where only increments by positive numbers are promised.}, output estimate $e$ satisfying $|e-\langle x, y\rangle| < \eps\cdot\|x\|_1\cdot\|y\|_1$.

\noindent In the above problems, we will be interested in the regime where $m\ll n$.

Our main result regarding a pseudo-deterministic algorithm for point query estimation is:
\begin{theorem} \label{thm:point-query}
    There is an $O\left(\frac{\log m}{\eps} + \log n\right)$-space pseudo-deterministic algorithm $\calA$ for point query estimation with the following precise guarantees.  For every sequence $s_1,\dots,s_m$ in $[n]^m$, there is a sequence $f_1',\dots,f_n'$ such that
    \begin{enumerate}
        \item For all $i$, $|f_i'-f_i| \le \eps m$ where $f_i$ is the frequency of $i$ in the stream.
        \item Except with probability $1/m$, for all $i\in[n]$ $\calA$ outputs $f_i'$ on query $i$.
    \end{enumerate}
\end{theorem}

We remark that the deterministic complexity of the problem is $\Omega(\frac{\log n}{\eps}$ (see Theorem \ref{detlbqeip}).

Towards establishing \pref{thm:point-query}, we recall two facts.
\begin{theorem}[Misra--Gries algorithm {\cite{misra1982finding}}]  \label{thm:misra-gries}
    Given a parameter $\eps$ and a length-$m$ stream of elements in $\{1,\dots,d\}$, there is a deterministic $O\left(\frac{\log d + \log m}{\eps}\right)$-space algorithm that given any query $s\in[d]$, outputs $f'_s$ such that $|f'_s-f_s|\le\eps m$ where $f_s$ is the number of occurrences of $s$ in the stream.  An additional guarantee that the algorithm satisfies is the following, which we call \emph{permutation invariance}.  Consider the stream
    \[
        s_1,s_2,\dots,s_m
    \]
    and for any permutation $\pi:[d]\rightarrow[d]$, consider the stream
    \[
        \pi(s_1),\pi(s_2),\dots,\pi(s_m).
    \]
    When the algorithm is given the first stream as input, let $f'_s$ denote its output on query $s$, and when the algorithm is given the second stream as input, let $g'_{\pi(s)}$ denote its output on query $\pi(s)$.  The algorithm has the guarantee that $f'_s = g'_{\pi(s)}$.
\end{theorem}

\begin{theorem}[Pairwise independent hashing, {\cite[Corollary 3.34]{vadhan2012pseudorandomness}}]   \label{thm:pairwise-indep}
    Assume $d\ll n$.  There is a pairwise independent hash function $h:[n]\rightarrow[d]$, which can be sampled using $O(\log n)$ random bits and also can be stored in $O(\log n)$ bits.
\end{theorem}

\begin{proof}[Proof of \pref{thm:point-query}]
    The algorithm is as follows.
    \begin{itemize}
        \item Sample a random pairwise independent hash function $h:[n]\to[m^3]$, which can be sampled and stored in $O(\log n)$ bits.
        \item Run the Misra--Gries algorithm with the following simulated stream as input:  for each $s$ streamed as input, stream $h(s)$ to the simulation.
        \item Given any query $s$, perform query $h(s)$ to the Misra--Gries algorithm running on the simulated stream, and return its output.
    \end{itemize}
    Let $S$ be the collection of elements of $[n]$ that occur in the input stream $s_1,\dots,s_m$.  Assuming $h$ maps $S$ into $[m^3]$ without any collisions\footnote{I.e. the restriction of $h$ to domain $S$ is an injective function.}, it follows from the permutation invariance property of the Misra--Gries algorithm from \pref{thm:misra-gries} the output of the above algorithm on any query $q$ is equal to $F(s_1,\dots,s_m,q)$ for a \emph{fixed} function $F$.  Thus if we show that $h$ indeed maps $S$ into $[m^3]$ injectively pseudo-determinism of the given algorithm would follow.
    
    Given $i,j\in S$, due to pairwise independence of $h$, the probability that $h(i)=h(j)$ is equal to $1/m^3$.  A union bound over all pairs of elements in $S$ tells us that $h$ is collision-free except with probability at most $1/m$, which implies that the above algorithm is indeed pseudo-deterministic.
\end{proof}

\begin{theorem} \label{thm:inner-prod-est}
    There is a (weakly) pseudo-deterministic algorithm for inner product estimation that uses $O\left(\frac{\log m}{\eps} + \log n\right)$ space.
\end{theorem}

The algorithm for inner product estimation is based on point query estimation, and towards stating the algorithm we first state a known result that helps relate the two problems.
\begin{lemma}[Easily extracted from the proof of {\cite[Theorem 1]{nelson2014deterministic}}]    \label{lem:inner-prod-from-point}
    Let $x,y,x',y'$ be vectors such that $\|x-x'\|_{\infty} \le \eps \|x\|_1$ and $\|y-y'\|_{\infty} \le \eps \|y\|_1$.  Now, let $x''$ (and respectively $y''$) denote $x'$ with everything except the maximum $1/\eps$ entries zeroed out.  Then the following holds:
    \[
        |\langle x'', y'' \rangle - \langle x, y\rangle| \le \eps\cdot\|x\|_1\cdot\|y\|_1.
    \]
\end{lemma}

\begin{proof}[Proof of \pref{thm:inner-prod-est}]
    Given a stream of updates to $x$ and $y$, run two instances of the point query estimation algorithm from \pref{thm:point-query} --- one for updates to $x$ and one for updates to $y$.  There are $x'$ and $y'$ that only depend on the stream such that
    \[
        \|x-x'\|_{\infty}\le\eps\cdot\|x\|_1 ~~~\text{and}~~~\|y-y'\|_{\infty}\le\eps\cdot\|y\|_1
    \]
    and except with probability $O(1/m)$ both point query algorithms respond to any query $i$ with $x'_i$ (and $y'_i$ respectively).  Maintaining these two instances takes $O\left(\frac{\log m}{\eps}+\log n\right)$ space.
    
    Next, enumerate over elements of $[n]$ and for each $i\in[n]$ query both instances with $i$, and store the running max-$1/\eps$ answers to queries to each instance along with the hashed identities of the indices of entries that are part of the running max.  Storing the running max takes $O\left(\frac{\log m}{\eps}\right)$ space, and storing a counter to enumerate over $[n]$ takes $\log n$ space.  Thus, at the end of this routine, except with probability $O(1/m)$ our two lists are equal to $(x'_{i_1},h(i_1)),\dots,(x'_{i_{1/\eps}},h(i_{1/\eps}))$ and $(y'_{j_1},h(j_1)),\dots,(y'_{j_{1/\eps}},h(j_{1/\eps}))$ respectively where $x'_{i_1},\dots,x_{i_{1/\eps}}$ are the max-$1/\eps$ entries of $x'$ and $y'_{j_1},\dots,y_{j_{1/\eps}}$ are the max-$1/\eps$ entries of $y'$.
    
    Finally, if there is $t,u$ such that $h(i_t)=h(i_u)$ or $h(j_t)=h(j_u)$, return `fail'; otherwise output \[\sum_{\ell\in\{h(i_t)\}_{t=1,\dots,1/\eps}\cap\{h(j_t)\}_{t=1,\dots,1/\eps}} x'_{\ell}y'_{\ell}.\]
    
    With probability at least $1-2/m$, the above quantity is equal to $\langle x'', y''\rangle$ from \pref{lem:inner-prod-from-point}, which lets us conclude via \pref{lem:inner-prod-from-point} that
    the output is within $\eps\cdot\|x\|_1\cdot\|y\|_1$ of the true inner product.
\end{proof}

Finally, we remark that the following lower bounds can be proved for deterministic algorithms.
\begin{theorem}\label{detlbqeip}
    Any deterministic algorithm for point query estimation and inner product estimation needs $\Omega\left(\frac{\log n}{\eps}\right)$ space.
\end{theorem}
\begin{proof}
    We prove a lower bound for point query estimation via a reduction from \Equality in communication complexity.  Alice encodes a $\log {n\choose 1/(3\eps)}$ bit string as a subset $S$ of $[n]$ of size $1/(3\eps)$ and runs the point query streaming algorithm on the input where she streams each element of this subset $3\eps m$ times.  She then sends the state of the algorithm to Bob, who can query every index in the universe and learn $S$ (the element corresponding to the query is in $S$ if and only if the response to the query is at least $2\eps\cdot m$), decode $S$ back to a $\log {n\choose 1/(3\eps)}$ and check if it is equal to his own input.  The space lower bound from the theorem statement then follows since $\log {n\choose 1/(3\eps)} = \Omega\left(\frac{\log n}{\eps}\right)$.
    
    A space lower bound for inner product estimation follows from the lower bound for point query estimation since the latter is a special case of the former when $x$ is the vector of frequencies and $y$ is a standard unit vector $e_i$ corresponding to query $i$. 
\end{proof}

\subsection{Retrieving a Basis of a Row-space}
We now work in a `mixed' model, where an input $n\times d$ matrix $A$ of rank-$\le k$ is given to us via a sequence of updates in a turnstile stream, and each entry at all times in the stream can be represented by an $O(\log n)$-bit word.  During this phase, there is an upper bound $T$ on the number of bits of space an algorithm is allowed to use.  In the ``second phase'', we are allowed to perform arbitrary computation and the goal is to output a basis for the row-span of $A$

We show a lower bound on $T$ of $\wt{\Omega}(nd)$ for deterministic algorithms, and a pseudo-deterministic algorithm that uses $\wt{O}(\poly(k)\cdot d)$ space in the streaming phase.

\begin{theorem}
    Any deterministic streaming algorithm for \RecBasis~needs $\wt{\Omega}(nd)$ space.
\end{theorem}
\begin{proof}
    Suppose the matrix $A$ is $0$, then the algorithm would have to output the empty set.  A $T$ space streaming algorithm for this problem could be used to solve the communication complexity problem of equality \Equality~using $T$ bits of communication.  In particular, Alice and Bob could encode their respective inputs $x$ and $y$ as matrices $M_x$ and $M_y$.  Alice can then run the $T$-space algorithm on adding $M_x$ in a turnstile stream, and send Bob the state of the algorithm.  Bob can then resume running the algorithm from Alice's state on updates that subtract $M_y$.  If Bob outputs the empty set, then $x = y$ and Bob outputs `yes'.  Otherwise, Bob outputs `no'.
\end{proof}

While the deterministic complexity is $\wt{\Omega}(nd)$, there is a pseudo-deterministic streaming algorithm which uses only $\wt{O}(\poly(k)+k\cdot d)$ in its streaming phase:

\begin{theorem} \label{thm:rec-basis}
    There is a pseudo-deterministic algorithm for \RecBasis~that uses $\wt{O}(\poly(k)+k\cdot d)$ space in its streaming phase, where the $\wt{O}(\cdot)$ hides factors of $\poly\log n$.
\end{theorem}

Towards giving a pseudo-deterministic algorithm, we first state a result about pseudorandom matrices that is a special case of \cite[Lemma 3.4]{CW09}.
\begin{theorem} \label{thm:pseudorandom-subspace-embed}
    There is a distribution $\calD$ over $m\times n$ matrices where $m=O(k\log n)$ with $\pm 1$ entries such that for any $n\times m$ matrix $U$ with orthonormal columns and $\bS\sim\calD$, the following holds with probability $1-1/\poly(n)$:
    \[
        \|U^T\bS\bS^TU-I\|_2\le 1/2.
    \]
    Further, the rows of $\bS$ are independent and each row can be generated by a $(k+\log n)$-wise independent hash family.
\end{theorem}

\begin{theorem}[$t$-wise independent hash families {\cite[Corollary 3.34]{vadhan2012pseudorandomness}}]
    There is a $t$-wise independent hash family $\mathcal{H}$ of functions from $[n]\to\{\pm 1\}$ such that sampling a uniformly random $h$ from $\mathcal{H}$ can be done using a $\poly(\log n, t)$-length random seed, and $h(x)$ for any $x\in[n]$ can be computed in $\poly(\log n, t)$ time and space from the random seed used to sample it.
\end{theorem}

As a consequence we have:
\begin{corollary}   \label{cor:rank-preserve}
    Let $A$ be a $n\times d$ matrix of rank $k$ and let $\calD$ be the distribution over $O(k\log n)\times n$ matrices from the statement of \pref{thm:pseudorandom-subspace-embed}.  Then, for $\bS\sim\calD$, $\bS A$ has rank $k$ with probability $1-1/\poly(n)$.
\end{corollary}
\begin{proof}
    We start by writing $A$ in its singular value decomposition $U\Sigma V^T$.  Since $A$ has rank $k$, $U$ is a $n\times k$ matrix with orthonormal columns and $\Sigma V^T$ surjectively maps $\mathbb{R}^d$ to $\mathbb{R}^k$.  From \pref{thm:pseudorandom-subspace-embed}, $\bS A$ is also full rank, which means the collection of vectors
    \[
        \{\bS Ax:x\in\mathbb{R}^d\}=\{\bS U\Sigma V^Tx:x\in\mathbb{R}^d\} = \{\bS Ux:x\in\mathbb{R}^k\}
    \]
    is a $k$-dimensional space, and hence $\bS A$ has rank $k$.
\end{proof}

\begin{proof}[Proof of \pref{thm:rec-basis}]
    Begin by sampling $\bS\sim\calD$ via a seed $\boldsymbol{s}$ of length $O(\poly(k)\cdot\poly\log(n))$ from which entries of $\bS$ can be efficiently computed where $\calD$ is the distribution over matrices given by \pref{cor:rank-preserve}, and maintain the sketch $\bS A$ in the stream.

    The row-span of $\bS A$ is exactly the same as that of $A$ assuming the two matrices have equal rank, which happens with probability $1-1/\poly(n)$.

    $\bS A$ is an $O(k)\times d$ matrix and each entry is a signed combination of at most $n$ entries of $A$ and hence there is a bit complexity bound of $\wt{O}(kd)$ on the space used to store $\bS A$.

    In the second phase (i.e., after the stream is over) of the algorithm, we first find an orthonormal basis $Q$ for the row-span of $\bS A$ and compute $\wt{\Pi}_A=QQ^T$.  And finally, use a deterministic algorithm to compute the singular value decomposition $\wt{U}\Sigma\wt{V}^T$ of $\wt{\Pi}_A$ and output the rows of $\wt{V}^T$.

    The row-span of $\bS A$ and $A$ are equal except with probability $1/\poly(n)$; assuming this happens, $\wt{\Pi}_A$ is exactly equal to $\Pi_A$, the unique projection matrix onto the row-span of $A$.  Write $\Pi_A$ in its singular value decomposition $U\Sigma V^T$.  If $\wt{\Pi}_A=\Pi_A$, $\wt{V}^T$ is exactly equal to $V^T$.  Since $V^T$ is given by a deterministic function of $A$, and the output of the algorithm $\wt{V}$ is equal to $V^T$ with high probability, our algorithm is pseudo-deterministic.
\end{proof}

\bibliographystyle{alpha}
\bibliography{main}

\end{document}